\documentclass[11pt]{article}
\usepackage[letterpaper,margin=1.00in]{geometry}
\usepackage{amsmath, amssymb, amsthm, amsfonts}
\usepackage{bbm}

\usepackage{ifthen}
\usepackage{tikz}
\usetikzlibrary{positioning,decorations.pathreplacing}

\usepackage{cite}
\usepackage{appendix}
\usepackage{graphicx}
\usepackage{color}
\usepackage{algorithm}
\usepackage[noend]{algpseudocode}
\usepackage{epstopdf}
\usepackage{wrapfig}
\usepackage{paralist}
\usepackage{wasysym}
\usepackage{svg}
\usepackage[textsize=tiny]{todonotes}

\usepackage{framed}
\usepackage[framemethod=tikz]{mdframed}
\usepackage[bottom]{footmisc}
\usepackage{enumitem}
\setitemize{noitemsep,topsep=3pt,parsep=3pt,partopsep=3pt}
\usepackage[font=small]{caption}
\usepackage{subcaption}
\usepackage{xspace}

\newtheorem{theorem}{Theorem}[section]
\newtheorem{lemma}[theorem]{Lemma}
\newtheorem{meta-theorem}[theorem]{Meta-Theorem}

\newtheorem{corollary}[theorem]{Corollary}

\newtheorem{observation}[theorem]{Observation}

\definecolor{darkgreen}{rgb}{0,0.5,0}
\usepackage{hyperref}
\hypersetup{
	unicode=false,          
	colorlinks=true,        
	linkcolor=red,          
	citecolor=darkgreen,        
	filecolor=magenta,      
	urlcolor=cyan           
}
\usepackage[capitalize, nameinlink]{cleveref}
\crefformat{footnote}{#2\footnotemark[#1]#3}
\crefname{theorem}{Theorem}{Theorems}
\Crefname{lemma}{Lemma}{Lemmas}
\Crefname{claim}{Claim}{Claims}
\Crefname{remark}{Remark}{Remarks}
\Crefname{observation}{Observation}{Observations}
\algnewcommand\algorithmicswitch{\textbf{switch}}
\algnewcommand\algorithmiccase{\textbf{case}}

\algdef{SE}[SWITCH]{Switch}{EndSwitch}[1]{\algorithmicswitch\ #1\ \algorithmicdo}{\algorithmicend\ \algorithmicswitch}%
\algdef{SE}[CASE]{Case}{EndCase}[1]{\algorithmiccase\ #1}{\algorithmicend\ \algorithmiccase}%
\algtext*{EndSwitch}%
\algtext*{EndCase}%

\newcommand{\eps}{\varepsilon}

\newcommand{\poly}{\operatorname{\text{{\rm poly}}}}

\renewcommand{\paragraph}[1]{\vspace{0.15cm}\noindent {\bf #1.}}

\newcommand{\FullOrShort}{full}
\ifthenelse{\equal{\FullOrShort}{full}}{
  
  \newcommand{\fullOnly}[1]{#1}
  \newcommand{\shortOnly}[1]{}

  }{

    \newcommand{\fullOnly}[1]{}
    \newcommand{\IncludePictures}[1]{}
   
  }

\begin{document}
\date{}
\title{Constant Approximation of Arboricity \\ in Near-Optimal Sublinear Time}
\author{
  Jiangqi Dai \\
  \small MIT \\
  \small jqdai@mit.edu
  \and
  Mohsen Ghaffari\\
  \small MIT \\
  \small ghaffari@mit.edu
  \and
  Julian Portmann \\
  \small ETH Zurich \\
  \small pjulian@ethz.ch
}

\maketitle

\begin{abstract}
We present a randomized algorithm that computes a constant approximation of a graph's arboricity, using $\tilde{O}(n/\lambda)$ queries to adjacency lists and in the same time bound. Here, $n$ and $\lambda$ denote the number of nodes and the graph's arboricity, respectively. The $\tilde{O}(n/\lambda)$ query complexity of our algorithm is nearly optimal. Our constant approximation settles a question of Eden, Mossel, and Ron [SODA'22], who achieved an $O(\log^2 n)$ approximation with the same query and time complexity and asked whether a better approximation can be achieved using near-optimal query complexity. 

A key technical challenge in the problem is due to recursive algorithms based on probabilistic samplings, each with a non-negligible error probability. In our case, many of the recursions invoked could have bad probabilistic samples and result in high query complexities. The particular difficulty is that those bad recursions are not easy or cheap to detect and discard. Our approach runs multiple recursions in parallel, to attenuate the error probability, using a careful \textit{scheduling mechanism} that manages the speed at which each of them progresses and makes our overall query complexity competitive with the single good recursion. We find this usage of parallelism and scheduling in a sublinear algorithm remarkable, and we are hopeful that similar ideas may find applications in a wider range of sublinear algorithms that rely on probabilistic recursions.
\end{abstract}
\thispagestyle{empty}

\bigskip
{
\newpage
\hypersetup{linkcolor=blue}
\tableofcontents
\thispagestyle{empty}
}
\newpage
\setcounter{page}{1}

\section{Introduction and Related Work}
This paper is centered on sublinear-time approximation algorithms for the highest density/sparsity in a graph, which captures a fundamental and important property of graphs. Several closely related definitions formalize the density measure, and they are all within a constant factor, as we outline next.
The \textit{arboricity} $\lambda(G)$ of a graph $G=(V, E)$ is defined as the minimum number of forests into which the edges $E$ can be partitioned.
By a well-known result of Nash-Williams~\cite{nash1961edge, nash1964decomposition, tutte1961problem}, an equivalent definition is $\lambda(G)=\max_{S\subseteq V} \{\lceil\frac{|E(S)|}{|S|-1}\rceil\}$, where $E(S)$ denotes the set of edges with both endpoints in $S$.
The \textit{maximum subgraph density} is defined as $\mathsf{dens}(G) = \max_{S\subseteq V} \{\frac{|E(S)|}{|S|}\}$.
One can see that $\mathsf{dens}(G) \leq \lambda(G) \leq \mathsf{dens}(G)+1$.
The degeneracy $\mathsf{degen}(G)$ is defined as the smallest value $k$ such that any induced subgraph of $G$ includes a vertex of degree at most $k$.
One can see that $\lambda(G) \leq \mathsf{degen}(G) \leq 2\lambda(G)-1$.
There are also other measurements, which are all within a constant factor of these bounds and have appeared under names such as the $k$-core number, width, linkage, thickness, Seekers–Wilf number, etc.
The arboricity will be the primary measure in our presentation of the results. Since all these measures are within a constant factor of arboricity, our approximation results apply to them as well.

Arboricity and the other density measures discussed above have numerous (algorithmic) applications; see, e.g., \cite{barenboim2010sublogarithmic,solomon2020improved, kothapalli2011distributed, bera2022counting,ghaffari2019distributed, chiba1985arboricity, brodal1999dynamic, bressan2021faster, bera2020linear}.
In particular, several sublinear algorithms assume that they receive upper bounds on the arboricity as a part of the input; see, e.g., \cite{eden2019arboricity, eden2019sublinear, eden2020faster,eden2022almost}.
We refer to \cite{eden2022approximating} for more on applications.

\medskip
\paragraph{Computing the arboricity}
It is well-known that the arboricity can be computed exactly in polynomial time. A classic exact algorithm due to Gabow runs in $\tilde{O}(m^{3/2})$ time~\cite{gabow1998algorithms}.
Furthermore, a $2$-approximation can be computed in $O(m+n)$ time by a simple algorithm of Matula and Beck~\cite{matula1983smallest}, which iteratively removes the lowest degree vertex until all vertices are removed and then reports the highest degree at the removal time (a bucket queue data structure is used to efficiently maintain the degrees). Such schemes are often called \textit{iterative peeling}. The approximation can be improved to $1+\epsilon$, for any constant $\eps>0$, with runtime $\tilde{O}(n)$, cf.\cite{king2022computing, mcgregor2015densest}.
In this work, our focus is on sublinear-time algorithms and we want to understand the minimal number of queries and time needed to determine or approximate graph's arboricity.

\subsection{Setup}
\paragraph{Model for sublinear-time algorithms} We consider the incidence-list query model: We assume vertices are numbered $\{1, 2, \dots, n\}$.
For any $v\in \{1, 2, \dots, n\}$ and any $i \in [1, n-1]$, the query $neighborQuery(v, i)$ returns the $i^{th}$ neighbor of the vertex $v$, if it has at least $i$ neighbors, and $\bot$ otherwise.
Notice that if the input graph is provided as one adjacency array for each vertex, which describes its neighbors, then this query can be performed with $O(1)$ cost.
We also assume access to a degree query $degreeQuery(v)$ which returns the degree of vertex $v$. We comment that it suffices to assume only the neighbor queries, if we ignore one logarithmic factor in the query complexity: one can answer the degree query for each node using $O(\log n)$ neighbor queries, via binary search.

\smallskip
\noindent \textbf{How many queries are needed to approximate the arboricity?}
Let us first discuss the lower bound.
Suppose there is a small but dense subgraph, which determines the arboricity.
For instance, consider a graph with $\Theta(n/\lambda)$ connected components, where one component is a clique of size $\lambda+1$, and all other components are graphs of arboricity somewhat below $\lambda$.
Informally, even to have a constant chance to pick a vertex in this clique, we need to query at least $\Theta(n/\lambda)$ vertices.
A formal lower bound appears in \cite{bhattacharya2015space} and (with some parameter adjustments, as done in ~\cite{eden2022approximating}) shows that for any $\poly(\log n)$-approximation (with at least $2/3$ success probability), we need $\tilde{\Omega}(n/\lambda)$ queries.
More generally, $k$-approximation needs $\Omega(n/(k\lambda))$ queries.
Hence, aiming for small approximation factors, the benchmark for us is query complexity $\tilde{O}(n/\lambda)$.

What is the best possible approximation factor within this $\tilde{O}(n/\lambda)$ query complexity bound?

\subsection{EMR, and our result}
In an elegant recent work, Eden, Mossel, and Ron~\cite{eden2022approximating} presented an $O(\log^2 n)$ approximation algorithm which has the near-optimal query complexity of $\tilde{O}(n/\lambda)$.
More concretely, their algorithm outputs a value $\hat{\lambda}$ such that $\lambda/200\log^2 n \leq \hat{\lambda} \leq \lambda$, with probability at least $1-1/n^2$.
They then posed this open question:
\begin{center}
	\begin{minipage}[t]{0.90\linewidth}
		``\textit{A natural question is whether the $O(\log^2 n)$ approximation factor can be improved with similar time complexity.}''
	\end{minipage}
\end{center}

\paragraph{Our result}
We show a simple algorithm that improves the approximation to $O(1)$ while retaining the near-optimal query complexity of $\tilde{O}(n/\lambda)$, and time complexity of $\tilde{O}(n/\lambda)$. 
\begin{theorem}
	\label[theorem]{thm:Main}
	There exists a constant $c\geq 2$ for which the following statement holds: There is a randomized algorithm in the incidence-list query model that, for any $n$-vertex graph $G=(V, E)$, outputs a value $\hat{\lambda}$ such that we have $\lambda/c \leq \hat{\lambda} \leq \lambda$, with high probability.\footnote{Throughout, and as standard, we use the phrase \textit{with high probability} (w.h.p.) to indicate that an event happens with probability at least $1-1/\poly(n)$, where the constant in the exponent of $\poly(n)$ can be made desirably large by adjusting other constants in the asymptotic notations.}
	Here, $\lambda$ denotes the arboricity of the graph $G$.
	The algorithm has query and time complexity $\tilde{O} (n / \lambda)$.
\end{theorem}
This constant approximation of arboricity in near-optimal query and time complexity essentially settles the question of Eden, Mossel, and Ron~\cite{eden2022approximating}, modulo figuring out the best possible constant. Our write-up does not attempt to optimize the approximation constant and prioritizes simplicity and readability. However, we think that even a fully optimized version of this algorithm would not approach an approximation close to $2$, and for that further ideas might be needed. 

\subsection{Method}
Here, we provide a bird's eye view of our approach and highlight how it differs from that of Eden et al.~\cite{eden2022approximating}. We note that this discussion will be at a very high level and will not enter the particular details, which are more nuanced.

\paragraph{Iterative peeling as a basic approach to arboricity} As mentioned above, a common ingredient in arboricity algorithms is a \textit{peeling mechanism}, which iteratively removes low-degree vertices in the subgraph induced by the remaining vertices. The most basic version, originally described by Matula and Beck~\cite{matula1983smallest}, is to remove each time the node with the minimum degree in the remaining graph, and eventually report the maximum degree at removal time as an approximation of the arboricity (one can see that this gives a $ 2$-approximation). If $\lambda$ is given as input as a hypothetical upper bound on the arboricity (which should be tested), this process can be alternatively described as iteratively removing vertices of degree at most $2\lambda$ in the remaining graph.

The base structure used in the algorithm of Eden, Mossel, and Ron~\cite{eden2022approximating} is a simple variant of the above, used earlier by Barenboim and Elkin in distributed algorithms~\cite{barenboim2010sublogarithmic}: In any graph with arboricity $\lambda$, one can partition the vertices into layers $L_1$, $L_2$, \dots, $L_{\log n}$ such that each node $v\in L_i$ has at most $4\lambda$ neighbors in \textit{the same or higher layers} $\bigcup_{j\geq i} L_j$. The process is simple: for $\log n$ iterations, in each iteration $i$ remove all nodes of degree at most $4\lambda$ from the graph and put them in $L_i$.  In graphs with higher arboricity, e.g., above $5\lambda$, there is an induced subgraph with minimum degree $5\lambda$ and thus some nodes cannot be put in such a layering. 

\paragraph{The approach of Eden et al.~\cite{eden2022approximating}, with cost estimations} 
The algorithm of Eden et al. devises an elegant scheme to simulate the $O(\log n)$-iteration peeling process in sublinear time. They build a query mechanism $Q_i(v)$ for each layer $i\in[1, \log n]$ which can determine, with certain approximate guarantees, whether node $v$ is in $\cup_{j\leq i} L_j$ or not. These queries will be applied on a sampling of the vertices in the graph, and if one of them is not in any layer, that is an indication that the graph has high arboricity. The query $Q_i(v)$ itself is built using the query mechanism of the lower layer $Q_{i-1}()$, invoked on a $(1/\lambda)$-probability sampling of the neighbors of $v$. The idea is that the number of sampled neighbors for which $Q_{i-1}(v)$ answers negative should be an estimator of the remaining degree of $v$ once nodes of $L_1$ to $L_{i-1}$ are removed, thus allowing to approximately determine the membership of $v$ in $L_{i}$. A challenge in their design is that calling $Q_{i-1}()$ on nodes that belong to higher layers $L_{i'}$ for $i'>i$ can be quite expensive (e.g., as their degrees can be much higher), and a priori it is unknown which neighbors of $v$ belong there. Eden et al. remedy this by computing/estimating the cost of $Q_{i-1}()$ for each sampled neighbor, in a cheaper way than running the procedure, and they \textit{prune} the $O(\log n)$ most costly neighbors. Notice that with high probability at most $O(\log n)$ neighbors in higher layers are sampled. In this scheme, when simulating each layer, we prune $O(\log n)$ (sampled) neighbors. Therefore, for $\Theta(\log n)$ layers of the simulation, we may prune up to $O(\log^2 n)$ (sampled) neighbors. In a rough sense, the scheme gets bottlenecked at a $\Theta(\log^2 n)$ approximation because of this.

\paragraph{Our approach, and scheduling parallel recursive calls} Our work is also based on the idea of simulating iterative peeling and testing whether all nodes can be peeled or not. Besides this basic point, our algorithm departs significantly from the approach of Eden et al. In particular, we do not use the $\Theta(\log n)$-layer structure mentioned above. More importantly, we have no mechanism to compute or estimate costs for nodes and we do not prune neighbors (while making the problem much more manageable technically, these estimations and prunings seem to be the source of approximation losses, in our understanding of the approach of Eden et al.). 

Our algorithm simulates iterative peeling using a careful probabilistic recursive structure that runs several recursions in parallel. Concretely, we will attempt to determine whether each node can be peeled or not by invoking several recursive peeling subroutines on random samplings of the neighbors (the exact recursions are more detailed). Crucially, each recursive process will have a non-negligible probability of having ``bad" samples and becoming costly. We do not have a mechanism to avoid, or significantly attenuate, this error probability. Instead, our approach relies on a recursion structure and careful scheduling rules that allow us to cope with recursions that have bad samples. The purpose of running many recursions in parallel is to have a high probability that at least one recursion has a ``good" sample and incurs only a small cost. However, it will not be known in advance which one of these multiple recursions is good, and the ``bad" ones may incur large complexities. As such, running these recursions sequentially one after the other, or running all of them naively in parallel at the same speed, will prove costly. 

We will use careful \textit{scheduling schemes} to remedy this. This is indeed a critical novelty in our paper. Our scheduling runs the multiple recursive calls in parallel to each other, with a very careful timing rule that determines the number of steps taken in each call at each moment of time. At a very high level, the scheduling staggers the recursions and takes exponentially more steps in the earlier recursive calls. In the analysis, we show how the overhead coming from running many recursive calls simultaneously is offset by considering the probabilities of the first recursive call that has a good probabilistic sample. 

Due to the nuanced nature of the algorithm, we do not attempt to summarize the details here. Fortunately, the algorithm is quite short and therefore we refer instead to the technical section.

\paragraph{Note on the scheduling} We emphasize that the careful scheduling mentioned above is critically necessary in our scheme. In particular, one cannot run the recursions one after the other, as the earlier ones might turn out to be the costly ones, while being hard to detect. Furthermore, even considering parallel invocations,
the simpler and far more natural idea of running the parallel recursions at the same time does not yield sublogarithmic approximations. Please see \Cref{subsec:loglogAlgorithm} for more on this, with a concrete example. We find it remarkable and unique that parallelism and scheduling rules are needed in such a seemingly inherent manner in a sublinear algorithm (where we do not have even the time to verify a solution otherwise). We are hopeful that our idea of parallel recursions with staggered scheduling might find applications in other sublinear algorithms, whenever one needs to invoke multiple \textit{probabilistic} recursive calls, some of which are erroneous and costly, and there are no cheap alternatives for identifying and terminating those costly recursions.

\paragraph{Roadmap} Even though the algorithm is short, its probabilistic recursive structure and especially its scheduling rule for parallel recursions can make the algorithm's design look unintuitive. To help accessibility, we first provide an $O(\log\log n)$ approximation algorithm in \Cref{sec:loglog} as a warm up and then present our $O(1)$ approximation algorithm in \Cref{sec:algorithm}.


\subsection{Other related work}
\paragraph{Tolerant testers for arboricity} Eden, Levi, and Ron~\cite{eden2020testing} study \emph{tolerant testing} of bounded arboricity. They give an algorithm to distinguish the following two cases in $\tilde{O} \left( \frac{n}{\gamma \sqrt{m}} + (1/\gamma)^{O(\log(1/\gamma))} \right)$ time: (A) graphs that are $\gamma$-close to having arboricity at most $\alpha$, i.e., graphs for which at most a $\gamma$-fraction of edges needs to be removed to achieve arboricity at most $\alpha$, and (B) graphs that are $20 \gamma$-far from having arboricity at most $3 \alpha$, i.e., graphs for which more than a $20 \gamma$-fraction of edges needs to be removed to achieve arboricity at most $3 \alpha$. Their algorithm can also compute a constant factor estimate of what they call the \emph{corrected arboricity}, which is the smallest arboricity of any graph that is $\gamma$-close to the input graph.
However, we note that this corrected arboricity does not provide an approximation of the (standard) arboricity, as also pointed out in \cite{eden2022approximating}.
For instance, a tree is $\gamma$-close to a graph with arboricity $\Omega(\sqrt{n})$, since by adding a $\gamma$-fraction of edges, we can create an $\Omega(\sqrt{n})$-size clique.

\paragraph{Prior work using scheduling in sequential algorithms} We would like to point out that there are classic works in complexity theory that use a scheduling idea of a general flavor similar to ours. In particular, Jones\cite{jones1997computability} shows that for any search problem, any optimal verification algorithm for it can be used to construct an optimal algorithm for it. Consider a search problem for which there exists an optimal verification algorithm with running time $\mathrm{Ver}(n)$. Jones' algorithm enumerates programs $P_1, P_2, \dots$ and executes them with the following schedule: for each $i \geq 1$, for every two operations performed by program $P_i$, program $P_{i+1}$ performs one operation. Whenever any program halts, its output is verified using the verifier. The overall algorithm terminates once the verifier confirms the correctness of some output. Let $P_C$ be the optimal program for the problem with running time $\mathrm{Opt}(n)$ and notice that $C$ is a pure constant (though presumably astronomically large, depending on the machine model and the listing of all programs). The overall algorithm has complexity at most $2^{C+1} \cdot (\mathrm{Opt}(n) + \mathrm{Ver}(n))$ steps. Since $C$ is a constant and $\mathrm{Ver}(n) = O(\mathrm{Opt}(n))$, this overall algorithm's running time is within a constant factor of $\mathrm{Opt}(n)$.

While the flavor of the scheduling is similar to ours, there are crucial differences, and our scheduling is striking due to its usage in sublinear (probabilistic and recursive) algorithms. In particular, Jones's algorithm is in no way applicable to our setting, as we are in a sublinear algorithms regime and there is no such verifier for our arboricity approximation problem. Indeed, building an approximate ``verifier"---a YES/NO tester whether arboricity is above or below a threshold, with some slack---is the whole point of our algorithm. Moreover, with Jones's algorithm, we do not have any new information about the actual complexity of the problem at hand; just that an asymptotically ``optimal" algorithm exists. Indeed, had it not been for this reason, the whole research field on algorithms (for search problems in which the certificate is easy to verify, which covers a vast range of the problems of interest) would be mute, as an asymptotically ``optimal" algorithm is already ``known". In contrast, we give a precise analysis and bound our algorithm's complexity (in particular, analyzing the overhead caused per recursion level, due to the scheduling and the probabilistic errors of the recursions).

\section{Preliminaries}
\label{sec:overview}

\paragraph{Notations} We work with a graph $G=(V, E)$, with $n=|V|$ vertices and $m=|E|$ edges, and we use $\lambda(G)$ to denote the arboricity of $G$. We use the symbol $\lambda$ as a free parameter in $[1, n]$, with the mindset that the most relevant problem is to compare $\lambda(G)$ and $\lambda$. For a subset $S\subseteq V$, we use $G[S]=(S, E_S)$ to denote the subgraph induced by $S$, i.e., the subgraph where we keep vertices in $S$ and only edges whose both endpoints are in $S$. For each node $v$, $N(v)$ denotes the set of its neighbors, and $deg(v)=|N(v)|$.

\medskip
\paragraph{Approximating via comparators} Approximation of arboricity boils down to building a \textit{comparator} of the graph's arboricity  $\lambda(G)$ against some given threshold $\lambda\in [1, n]$. In a $c$-approximate comparator against $\lambda$, we want the comparator to say \textit{YES} if $\lambda(G) \leq \lambda$ and say \textit{NO} if $\lambda(G)\geq c\lambda$. In other cases, we do not need a guarantee from the comparison's output. In our algorithm, we build this comparator using $\tilde{O}(n/\lambda)$ queries. Then, the approximation algorithm follows by just running the comparator for $\lambda=n/2^i$ for $i=0, 1, 2, \ldots$ until we reach the first value of $\lambda$ such that the comparator outputs YES. This gives a $2c$ approximation of $\lambda(G)$ using $\tilde{O}(n/\lambda(G))$ queries. 

A basic approach for this comparator (which can be done in roughly linear time in the graph size) is via iteratively peeling low-degree vertices, as we recall next.

\medskip
\paragraph{Iterative peeling} Iteratively and sequentially remove vertices of degree at most $2\lambda$ from the graph, along with all of their edges, forming a sequence $X=(x_1, x_2, x_3, \dots, x_k)$ of removed nodes, ordered according to the removal time, until no further node can be removed. Here $k$ is the total number of nodes that got removed when the process stopped. We comment that, for simplicity, we assume that $2\lambda$ is an integer. The $\pm 1$ rounding loss caused when this is not satisfied does not invalidate the constant approximation. 
\begin{observation} Suppose we orient edges incident on nodes in $X$ toward the endpoint with the higher index if both endpoints are in $X$ (i.e., from $x_i$ to $x_j$ where $j>i$), and toward the endpoint in $V\setminus X$ if one endpoint is outside $X$. Then each node in $X$ has outdegree at most $2\lambda.$
\end{observation}

\begin{observation}
If $\lambda(G)\leq \lambda$, we would have $k=n$ meaning that all nodes get removed. If $\lambda(G)>2\lambda$, then we would have $k\leq n-2\lambda-1<n$.    
\end{observation}
\begin{proof}
 First, we argue that if $\lambda(G)\leq \lambda$, then $k=n$. The reason is that for any nonempty set $S\subseteq V$ of the remaining nodes, the induced subgraph $G[S]$ has at most $\lambda (|S|-1)$ edges since $\lambda(G)\leq \lambda$. Hence, there is a node of degree at most $2\lambda$ in $G[S]$ that gets removed. 
 
 Second, we argue that if $\lambda(G)>2\lambda$, we would have $k\leq n-2\lambda-1<n$. The reason is as follows: Recall that by results of Nash-Williams\cite{nash1961edge}, we have $\lambda(G)=\max_{S\subseteq V} \{\lceil\frac{|E(S)|}{|S|-1}\rceil\}$. If $\lambda(G)>2\lambda$, considering the integrality of $\lambda(G)$, there exists a subset $S\subseteq V$ such that $\lceil\frac{|E(S)|}{|S|-1}\rceil\geq 2\lambda+1$. Let $U$ be a \textit{minimal} such subset and notice that $G[U]$ has minimum degree at least $2\lambda+1$ (if there was a node with degree at most $2\lambda$ in $G[U]$, removing it and its edges would give a smaller subgraph that still satisfies the inequality). Notice that $G[U]$ is non-empty, as $|E(U)|$ is an integer that satisfied $|E(U)|> 2\lambda(|U|-1)\geq 0$. Since $G[U]$ is non-empty and has minimum degree exceeding $2\lambda$, we know $|U|\geq 2\lambda+1$. Because of its minimum degree, no node of $U$ ever gets removed in the iterative peeling process, and thus $k=|X|\leq k\leq n-2\lambda-1<n$.
 \end{proof}

\paragraph{Core subgraph} In the case $\lambda(G)>2\lambda$, as argued in the above proof, the graph has a nonempty induced subgraph $G[U]$ with minimum degree above $2\lambda$. We call $G[U]$ the \textit{core} of $G$.

\smallskip

\paragraph{Implementing sampling with efficient time complexity}
In our algorithms, to have a good time complexity (besides the query complexity), we need an efficient implementation of the process of sampling a subset $B$ of a given set $A$, where each element is sampled with probability $p$. We want the time complexity to be proportional to the generated sample size $|B|$, with expectation $p|A|$, rather than the ground set size $|A|$. The following lemma provides such an implementation. We comment that this is a basic and simple question, and perhaps known in prior work. We did not know a citation for it and thus decided to write a proof.

\begin{lemma}
\label[lemma]{lem:sampling}
    Suppose we are given a set of $k$ elements $A$ and we want to create a sample subset $B\subseteq A$ by including each element of $A$ in $B$ with probability $p\in [0,1]$. We can do this in time $O(|B|)$. In particular, the expected time complexity to generate $B$ is $O(pk)$.
\end{lemma}
\begin{proof} We will first sample the value $t=|B|$ in $O(t)$ complexity, according to the right distribution, and then draw $t$ samples from $A$. 

Notice that $t$ follows a binomial distribution. We should set $t$ equal to $0$ with probability $q_0 = \binom{k}{0}(1 - p)^k$, set it equal $1$ with probability $q_1 = \binom{k}{1}(1 - p)^{k - 1}p^1$, etc. We can calculate $q_i$ based on $q_{i-1}$ in constant time, via $q_i=\frac{p(k - i + 1)}{(1 - p)i} q_{i - 1}$. To generate $t$, we do not need to compute all $q_i$ at once; instead, we interleave the computation of $q_{i}$ with the sampling. That is, we try $i = 0, 1, 2, \cdots$, calculate $q_i$ and check sequentially whether or not we will set $t$ to be $i$. The entire process takes $O(t)$ time, where $t$ is the eventual selected sample size.

Next, we describe how to draw $t$ samples from $A$ uniformly at random in time $O(t)$: Consider an array $a_1, \cdots, a_k$ standing for the elements in $A$. For $i = 1, \cdots, t$, we swap $a_i$ and $a_j$, where $j$ is uniformly generated number in the range $[i, k]$. In this way, each combination of the tuple $(a_1, \cdots, a_t)$ will be generated with equal probability. Since we only need to maintain the positions of the array that has been queried at least once, we will need to access at most $2t$ locations in $A$ (whose indices are determined independent of the values in $A$) and a hash table suffices to make the process run in $O(t)$ time. Also, since $\mathbb{E}(t) = kp$, the expected total time complexity is $O(kp)$.   
\end{proof}


\section{Warm up: An $O(\log\log n)$ Approximation Algorithm}
\label{sec:loglog}
\subsection{Algorithm} 
\label{subsec:loglogAlgorithm}
The base module is an approximate \textit{comparator} of the graph's arboricity $\lambda(G)$ against any given threshold  $\lambda \in [1,n]$, with the following guarantee: if $\lambda(G)\leq \lambda$, then the output is YES, with high probability. If $\lambda(G)\geq \lambda \cdot C\log\log n$, for some constant $C$, then the output is NO, with high probability. In other cases, there is no guarantee.

This comparator consists of $O(\log n)$ independent \textit{tests} each with a YES/NO result. The comparator output is set equal to the \textit{majority} of these test results. Each test works as follows: we randomly sample $\Theta(n\log n/\lambda)$ nodes and perform a process called try-to-peel($v$) on each sampled node $v$. We discuss this process soon; the main property is whether it terminates or not (using a small number of queries). If \textit{all} processes in a test terminate before we have used $\tilde{\Theta}(n/\lambda)$ queries overall, then the test result is YES. Otherwise, it is NO. We next described this process.

\begin{center}
\begin{mdframed}
\paragraph{The \textit{try-to-peel($v$)} process} When this process is called on vertex $v$, we create $L=\Theta(\log n)$ new \textit{procedures} $H_1(v)$, $H_{2}(v)$, \dots, $H_L(v)$, which we run in parallel, with a certain scheduling rule to be discussed. In each procedure $H_{i}(v)$, we sample each neighbor $u$ of $v$ with probability $1/(10\lambda)$. If no neighbor is sampled, the proceduce $H_{t}(v)$ is considered terminated. Otherwise, the procedure $H_{t}(v)$ invokes a recursion try-to-peel($u$) on each sampled neighbor $u$. All the random samplings in different procedures use independent randomness (and the sampling is always done with respect to the original graph $G$). We call the procedure $H_{t}(v)$, for $t\in [L]$, \textit{terminated} once all the recursions it invoked on the sampled neighbors have terminated. Furthermore, we consider the try-to-peel($v$) process \textit{terminated} as soon as one of these $L=O(\log n)$ procedures terminates. We discuss how these $L$ processes $H_1(v)$, $H_{2}(v)$, \dots, $H_L(v)$ are scheduled.

\paragraph{Scheduling/pacing different procedures} In the \textit{try-to-peel($v$)} process, we have a counter $y = 1, 2, 3, \cdots$, which we will increment one by one. When the counter is $y$, let $zeros(y)$ be the number of trailing zeros in the binary representation of $y$. Then, at counter $y$, we take one step in procedures $H_{t}(v)$ for each $t\in [1, zeros(y)+1]$. That is, through the course of try-to-peel($v$), we take one step in procedure $H_{1}(v)$ with every $2^0=1$ increase in $y$, one step in procedure $H_{2}(v)$ with every $2^{1}=2$ increase in $y$, and generally one step in procedure $H_{t}(v)$ with every $2^{t-1}$ increase in $y$.
\end{mdframed}
\end{center}

\paragraph{Why we need a non-trivial scheduling scheme} Here, we provide an intuitive discussion on why we need careful scheduling schemes, and why simpler ideas (e.g., many independent repetitions run simultaneously) do not work. 
Fix a vertex $v \in V(G)$. For simplicity, assume that $\Pr[\textit{$H_t(v)$ terminates}] = 0.8$ and $\mathbb{E}[\textit{query complexity of $H_t(v)$} | \textit{$H_t(v)$ terminates}] = M$. Then, with our scheduling scheme, the query complexity of $\textit{try-to-peal(v)}$ is bounded by $4M$, as illustrated in \Cref{lem:warmup-smallCase}. In contrast, suppose that we schedule naively, i.e., take one step in each procedure $H_t(v)$ with every increase in the counter $y$. Now, the expected query complexity of $\textit{try-to-peal(v)}$ would be $L \cdot \mathbb{E}[\min(T(H_1(v)), \cdots, T(H_L(v)))]$, where $T(H_t(v))$ is a random variable denoting the number of queries $H_t(v)$ makes and satisfies $\mathbb{E}[T(H_t(v))] = M$. One can see that $L \cdot \mathbb{E}[\min(T(H_1(v)), \cdots, T(H_L(v)))]$ can be much greater than $4M$ (e.g., when $T(H_t(v))$ is concentrated around its mean). Notice that while this single $L$ factor might appear harmless as it is only $O(\log n)$, this factor will appear in every level of recursion, and the multiplication of these factors blows up. In contrast, we are able to attenuate and control the effect of the $4$ factor in the $4M$ bound, with the constant in our $1/(10\lambda)$ sampling probability; see \Cref{pointedInequality} in the proof of \Cref{lem:warmup-smallCase}. Because of this blow up of $L$ factor, which appears in every level of recursion, the simpler scheduling does not work, and would not achieve any $\tilde{O}(n/\lambda)$ query complexity\footnote{Let us provide an example graph and a back-of-envelop estimation showing that the naive scheduling of running all $\Theta(\log n)$ copies in parallel would not work with a small query complexity. However, working out a full example requires more detailed probabilistic statements and is beyond the scope of this paper. Consider a graph made of layers $L_1, L_2, L_3, ...,$ where each node in layer $L_i$ has $\lambda$ neighbors in layer $L_{i+1}$ and each node in layer $L_{i+1}$ has $\Theta(\lambda \log\log n)$ neighbors in layer $L_i$. So, we have $n/(\Theta(\log\log n)^i)$ nodes in layer $i$. When we do try-to-peel($v$) on a node $v$ in layer $j$ and invoke $\Theta(\log n)$ procedures each on an independent sampling of the neighbors with probability $1/(10\lambda)$, we expect each procedure to sample $\Theta(\log\log n)$ neighbors in layer $j-1$, and indeed that will happen for all of the $\Theta(\log n)$ copies with probability at least $1-1/\poly(\log n)$. We run all the $\Theta(\log n)$ instances in parallel at the same speed. Even ignoring samples from higher layers, the expected cost of simulating a node in layer $j$ would be $\Theta(\log n)$ times more than the expected cost of simulating $\Theta(\log\log n)$ nodes in layer $j-1$ (and one can see stronger concentrations in the query complexity, besides just this expected bound, considering the independence of the samplings). Working out this recursion, the cost of a node in layer $j$ would be $(\Theta(\log n \cdot \log\log n))^j$. Since there are $n/(\Theta(\log\log n)^j)$ nodes in layer $j$, the total expected query complexity would exceed $\Theta(n)$.}.

\bigskip
\noindent \textbf{Remark on implementation details.} To keep the time complexity of our algorithm small, it is important that we implement the process of \textit{sampling} neighbors efficiently, in such a way that the time to produce the sample is asymptotically equal to the number of samples, rather than the size of the set from which we are sampling. We describe such an implementation in a generic sampling setting in \Cref{lem:sampling}.


\subsection{Analysis}

\paragraph{Outline of the analysis} The analysis involves considering two cases for the test, the low arboricity case where $\lambda(G)\leq \lambda$ and the high arboricity case where $\lambda(G)\geq C\log\log n \cdot \lambda$. 
Our analysis argues that if $\lambda(G)\leq \lambda$, then in each test, with probability at least $2/3$, the result will be YES, i.e., all processes will terminate within query complexity budget $\tilde{O}(n/\lambda)$. In contrast, if $\lambda(G)\geq c\log\log n \cdot \lambda$, we argue that in each test, with probability at least $1-1/\poly(n)$, the result will be NO and at least one sampled vertex $v$ will have its process try-to-peel($v$) not terminate. Considering that the output of the comparator is that of the majority of these $O(\log n)$ independent tests, we get that with high probability, if $\lambda(G)\leq \lambda$ the output is YES, and if $\lambda(G)\geq \lambda \cdot c\log\log n$ the output is NO. We next discuss each of these two cases separately.

\medskip
\subsubsection{Low-Arboricity Case} Suppose that the graph $G$ has arboricity $\lambda(G)\leq \lambda$. We first discuss the analysis outline and intuition. The formal analysis is provided afterward in \Cref{lem:warmup-smallCase} and \Cref{crl:warmup-smallCase}. We show in \Cref{lem:warmup-smallCase} that if $\lambda(G)\leq \lambda$, the try-to-peel($x_i$) process terminates with probability at least $1-1/\poly(n)$ for each node (and with a small query complexity). We will argue this by induction, for the sequence $X=(x_1, \dots, x_n)$ of nodes removed in iterative peeling (cf. \Cref{sec:overview}). In each $H_t(x_i)$ procedure, with probability at least $1-2\lambda/10\lambda=0.8$, we get a \textit{good sample} in the sense that node $x_i$ will have no sampled neighbor $x_{j'}$ for $j'>i$. If that happens, by induction, $H_{t}(x_i)$ will terminate (with a decently bounded query complexity, to be discussed). Since each procedure has $0.8$ probability of having a good sample, with probability at least $1-1/\poly(n)$, one of the $L=\Theta(\log n)$ procedures will have a good sample and terminate, and thus try-to-peel($v$) will terminate. 

Furthermore, we show that in this $\lambda(G)\leq \lambda$ case, the expected total query complexity of try-to-peel($x_i$) among all sampled nodes is $\tilde{O}(n/\lambda)$. Thus the test terminates within the $\tilde{O}(n/\lambda)$ query budget, with probability at least $2/3$. For the former, we prove in \Cref{lem:warmup-smallCase} that when we call try-to-peel($x_i$) on each node $x_i$, the process terminates in expected time $T(x_i)$, such that we have $\sum_{i \in X} T(x_i) = O(n)$. Since in the test we sample each node with probability $\tilde{O}(1/\lambda)$ and run the processes try-to-peel($x_i$) only on the sampled nodes, these processes will terminate within a total expected query complexity of $\tilde{O}(n/\lambda)$, and thus with probability at least $2/3$ within the budget query complexity of $\tilde{O}(n/\lambda)$, as we will conclude in \Cref{crl:warmup-smallCase}. 

We next present the formal statement. 
 
\begin{lemma}\label[lemma]{lem:warmup-smallCase}
      Suppose that $G$ is a graph with arboricity at most $\lambda$, and let $X = (x_1, \cdots, x_n)$ be the order of the vertices according to the time of removal in the sequential peeling process. Then, when we call try-to-peel($x_i$) on each node $x_i$, the process terminates with probability $1-1/\poly(n)$ within $O(n)$ time, and more importantly in expected time $T(x_i)$ such that we have $\sum_{x_i \in X} T(x_i) = O(n)$. 
\end{lemma}
\begin{proof}
Consider try-to-peel($x_i$) and its $L=O(\log n)$ independent procedures  $H_1(x_i)$, $H_{2}(x_i)$, \dots, $H_L(x_i)$. As discussed above, let us say procedure $H_{t}(x_i)$ has a good sample if in it node $x_i$ has no sampled neighbor $x_{j'}$ for $j'>i$. Since node $x_i$ has at most $2\lambda$ such neighbors and each is sampled with probability $1/(10\lambda)$, the probability that the procedure has a good sample is at least $1-2\lambda/10\lambda=0.8$. Let $t^*$ be the smallest index $t$ such that $H_{t}(x_i)$ has a good sample. Notice that, with probability at least $1-0.2^{L} = 1-1/\poly(n)$, the index $t^*$ is well defined and $t^*\leq L$. For any $r\geq 1$, the probability that $t^*=r$ is at most $(0.2)^{r-1}$. Let us add one comment about the former point: Technically, we apply also a union bound over all the at most $O(n)$ calls to try-to-peel($x_j$) originating from try-to-peel($x_i$) to assume that $t^*$ is well defined in each such triggered try-to-peel($x_j$). Considering the first $\Theta(n)$ such calls suffices for all possibilities of termination in $O(n)$ time. Thus we can union bound the $1/\poly(n)$ failure probability over all the at most $O(n)$ such calls, to conclude that each of them has a well-defined index $t^*\leq L$ for a good sample, with probability $1-O(n)/\poly(n)=1-1/\poly(n).$

Now let us examine the query complexity, conditioned on $t^*=r$. We first consider procedure $H_{r}(x_i)$, which has a good sample. This procedure terminates once all its sampled neighbors $x_j$ for $j<i$ have terminated. For brevity, we use the phrase $Comp(.)$ to indicate the query complexity of the subroutine in the paranthesis, and we let $\mathcal{E}_{r}(x_i)$ denote the event that $H_{r}(x_i)$ has a good sample. Therefore, we have
\begin{align*}
    &\mathbb{E}[Comp(H_{r}(x_i))| \mathcal{E}_{r}(x_i)]  
    \\ &\leq \frac{1}{10\lambda} \sum_{j < i, (x_j, x_i) \in E(G)} \mathbb{E}[Comp(\textit{try-to-peel}(x_j))] \\ 
    &+ \Theta(deg(x_i)/\lambda).
\end{align*}

Suppose that procedure $H_{r}(x_i)$ terminates once the counter in try-to-peel($x_i$) has reached value $y$. At that point, each procedure $H_{t'}(x_i)$ has taken $y/(2^{t'-1})$ steps. So the total number of steps taken by all procedures at the moment that we terminate {try-to-peel($x_i$)} is bounded by $2y$, which is at most $2^{r}$ times the number of steps taken by procedure $H_{r}(x_i)$. For brevity, we use the phrase $Comp(.)$ to indicate the query complexity of the subroutine in the paranthesis, and we let $\mathcal{E}_{r}(x_i)$ denote the event that $H_{r}(x_i)$ has a good sample. Hence,

\begin{alignat}{2}
&&&\mathbb{E}[Comp(\textit{try-to-peel}(x_i))] \nonumber\\
&\le &&\sum_{r=1}^{L} Pr[t^*=r]\,2^r\,
   \mathbb{E}[Comp(H_{r}(x_i))\mid \mathcal{E}_{r}(x_i)] \nonumber\\
&\le &&\sum_{r=1}^{L} 0.2^{r-1}\,2^r\,
   \mathbb{E}[Comp(H_{r}(x_i))\mid \mathcal{E}_{r}(x_i)] \nonumber\\
&\le &&4\,\mathbb{E}[Comp(H_{r}(x_i))\mid \mathcal{E}_{r}(x_i)] \nonumber\\
&\le &&4\, \Big(\frac{1}{10\lambda}
   \sum_{\substack{j<i \\ (x_j,x_i)\in E(G)}}\!
   \mathbb{E}[Comp(\textit{try-to-peel}(x_j))] \tag{*}\label{pointedInequality} + \Theta(\tfrac{\deg(x_i)}{\lambda})\Big) \nonumber 
\end{alignat}

Now, since each node $x_j$ has at most $2\lambda$ neighbors $x_i$ for $i>j$, the contributions of $deg(x_j)/\lambda$ for each node $x_j$ in the summation $\sum_{i=1}^{n}\mathbb{E}[Comp(\textit{query try-to-peel}(x_i))]$ make a geometric series with ratio $2\lambda \cdot \frac{4}{10\lambda} = 0.8$. A detailed analog of this point is presented later in our proof of \Cref{lem:recursiveT}, as used in \Cref{lem:smallCase}. As a result, one can conclude that 
\begin{align*}\sum_{i=1}^{n} \mathbb{E}[Comp(\textit{try-to-peel}(x_i))] = \sum_{i=1}^{n} \Theta(deg(x_i)/\lambda) = O(n).\end{align*}
\end{proof}
\begin{corollary}
    \label[corollary]{crl:warmup-smallCase}
    If $\lambda(G)\leq \lambda$, then in each test, with probability at least $2/3$, the result will be YES. That is, all processes will terminate within query complexity budget $\tilde{O}(n/\lambda)$.
\end{corollary}
\begin{proof}
\Cref{lem:warmup-smallCase} shows that try-to-peel($x_i$) on each node $x_i$ terminates with probability $1-1/\poly(n)$ and in expected time (and thus also query complexity) at most $T(x_i)$, such that we have $\sum_{i \in X} T(x_i) = O(n)$. Our test runs try-to-peel($x_i$) on only a sample of $\tilde{O}(n/\lambda)$ nodes, by sampling each node with probability $\tilde{O}(1/\lambda)$. Hence, the expected time query complexity of running the test is $\tilde{O}(n/\lambda)$. Hence, the test terminates all the try-to-peel($x_i$) runs within the $\tilde{O}(n/\lambda)$ query budget---where we have increased the constant in the $\tilde{O}()$ notation by a $3$ factor---with probability at least $2/3$.
\end{proof}

\subsubsection{High-Arboricity Case} 
Suppose that $\lambda(G)\geq C\log\log n \cdot\lambda$. We first discuss the analysis outline. The formal analysis is provided afterward in \Cref{lem:warmup-largeCase} and \Cref{crl:warmup-largeCase}. Consider the core induced subgraph $G[U]$ which has minimum degree at least $C\log\log n \cdot\lambda$. We argue that nodes of the core are not expected to terminate. More concretely, we show that when running try-to-peel($v$) for a node $v\in U$, with probability at least $1-1/\poly(\log n)$, each of the $L$ procedures has at least $\log\log n$ sampled neighbors from the core $U$ and because of this, the probability of termination for try-to-peel($v$) is at most $1/\poly(\log n)$. Hence, given that we sample $\Theta(n\log n/\lambda)$ nodes, we will have $\Omega(\log n)$ nodes sampled from the core with high probability and at least one of them will not terminate, with high probability.  

We next present the formal statement.
 
\begin{lemma}\label[lemma]{lem:warmup-largeCase} There exists a constant $C\geq 100$ for which the following holds: Suppose that $G$ is a graph with arboricity at least $C\log\log n \cdot \lambda$. Let $G[U]$ be an induced subgraph with a minimum degree at least $C\log\log n \cdot \lambda$. For any $v \in U$,  with probability at least $1-1/\log^3 n$, the process try-to-peel($v$) will not terminate. 
\end{lemma} 
\begin{proof}[Proof Sketch]
Consider the process try-to-peel($v$) and the $L=O(\log n)$ procedures in it $H_{1}(v)$, $H_{2}(v)$, ..., $H_{L}(v)$. Let us call the sampling of procedure $H_{t}(v)$ \textit{bad} if less than $\log\log n$ neighbors of $v$ in $U$ get sampled. Since node $v$ has $C\log\log n \cdot \lambda$ neighbors in $U$, and the sampling probability is $1/(10\lambda)$, the expected number of sampled neighbors is at least $10\log\log n$. Thus, with probability at least $1-1/\log^5{n}$, at least $\log\log n$ neighbors are sampled in the procedure  $H_{t}(v)$. Let us say process try-to-peel($v$) is bad if one of the $L$ procedures has a bad sample. Thus, the probability of the process being bad is at most $O(\log n)/\log^5 n$. That is, with probability $1-O(1/\log^4 n)$, the process try-to-peel($v$) is not bad. In this case, since each procedure has $\log\log n$ sampled nodes, and each sampled node terminates with probability at most $1/\log^3 n$, all with independence, we can conclude that the probability that any of them terminate is at most $(1/\log^3 n)^{\log\log n} \cdot O(\log n).$ Hence, with probability $1-O(1/\log^4 n) - (1/\log^3 n)^{\log\log n} \cdot O(\log n) \geq 1-1/\log^3 n,$ the process try-to-peel($v$) does not terminate.\footnote{The argument here is not fully formal but it can be formalized with an induction based on the height of the recursion tree. We keep that formal argument for the variant of this result in \Cref{lem:largeCase}.} 
\end{proof}

\begin{corollary}
    \label[corollary]{crl:warmup-largeCase}
    If $\lambda(G)\geq C\log\log n \cdot \lambda$, then in each test, w.h.p., the result will be NO, i.e., some processe will not terminate.
\end{corollary}
\begin{proof} The test samples $\Theta(n\log n/\lambda)$ nodes. Since the core has at least $C\log\log n \cdot \lambda$ vertices, we will have $\Omega(\log n)$ nodes sampled from the core $G[U]$ with high probability. For each core node $v\in U$, by \Cref{lem:warmup-largeCase}, with probability at least $1-1/\log^3 n$, the process try-to-peel($v$) will not terminate. Hence, with high probability, at least one of the $\Omega(\log n)$ processes triggered on core nodes will not terminate.
\end{proof}

\section{Main Result: An $O(1)$ Approximation Algorithm}
\label{sec:algorithm}
Our algorithm is surprisingly short. However, it is based on a complex probabilistic recursive structure, with careful scheduling and decision rules meant to control the error probabilities of different directions. Because of this, the algorithm's design might be less intuitive. For conciseness, we first present the algorithm in its entirety in \Cref{subsec:mainAlgorithm}. We then provide additional discussions in the \Cref{subsec:discussion}, and then the analysis in \Cref{subsec:analysis}.

\subsection{Algorithm}\label{subsec:mainAlgorithm}
The algorithm involves a general outline, mostly similar to that of \Cref{sec:loglog}, and two key main probabilistic recursive modules which are quite different than before and nuanced. We first describe the outline to clarify the role of the modules, and then present the modules.

\subsubsection{Outline} As before, the main part of the algorithm is a comparator that, given threshold $\lambda$ runs in $\tilde{O}(n/\lambda)$ time and generates a YES/NO output, with the following interpretation: YES indicates low arboricity and NO indicates high arboricity. The comparator consists of $O(\log n)$ independent \textit{tests}, each with a YES/NO result, and the comparator outputs the majority decision among these results. Each test works as follows: we sample each vertex with $\frac{1}{\lambda}$ probability independently, and we run the \textit{fortified peeling procedure} $F(s)$, which is described in \Cref{subsub:modules}, on each sampled vertex $s$. If all of these procedures finish in a total of $\frac{10 S}{\lambda}$ time steps, the test result is YES. Here, $S = \tilde O(n)$ and the exact value of it will be discussed later in the analysis. Otherwise, the test result is NO.  

By definition, the time complexity and query complexity of each test, and thus also the comparator, is $\tilde{O}(n/\lambda)$. In the analysis, we will argue that, w.h.p., if $\lambda(G)\leq \lambda$, the comparator will say YES and if $\lambda(G)>c\lambda$, it will say NO. Given this, the constant approximation algorithm follows easily as described later in the proof of \Cref{thm:Main} presented in \Cref{subsub:wrap}.

\subsubsection{Main Modules}
\label{subsub:modules} The core ingredients in our algorithm are two probabilistic recursive modules/processes: a basic process named \textit{single-shot peeling process $H(v)$}, and a more complex process \textit{fortified peeling process $F(v)$} which makes use of many parallel invocations of $H(v)$. We next present these two modules. 

\medskip
\begin{center}
\begin{minipage}{0.9\linewidth}
\begin{framed}
\paragraph{Single-shot peeling process $H(v)$} Sample a subset $\Gamma(v)\subseteq N(v)$ by including each neighbor $u$ of $v$ in $\Gamma(v)$ independently with probability $\frac{1}{25 \lambda}$. Then, for each neighbor $u\in \Gamma(v)$, run the fortified peeling process $F(u)$. If/once all these processes $F(u)$ have terminated (or if no neighbor was sampled), we call process $H(v)$ terminated.
\end{framed}
\end{minipage}
\end{center}

\noindent \textbf{Remark.} As we also saw in the warm-up algorithm of \Cref{sec:loglog}, to keep the time complexity of our algorithm small, it is important that we implement the procedure of \textit{sampling} neighbors efficiently, in such a way that the time to produce the sample is asymptotically equal to the number of nodes in the generated sampled. See \Cref{lem:sampling} for such an implementation.


  \begin{center}
\begin{minipage}{0.9\linewidth}
\begin{framed}
\paragraph{Fortified peeling process $F(v)$} The process $F(v)$ involves an infinite number of \textit{batches}, which are run in parallel with a careful scheduling rule to be discussed. The process $F(v)$ is called \textit{terminated} as soon as one of these batches is terminated. The $t$-th batch consists of $2t - 1$ independent recursive calls to the single-shot peeling $H(v)$---denoted by $H_{t, 1}(v), H_{t, 2}(v), \cdots, H_{t, 2t - 1}(v)$---which run at the same speed and in parallel. These calls use independent randomness. We call batch $t$ \textit{terminated} iff more than half of the processes $H_{t, .}(v)$ of this batch terminate. 

\smallskip
\paragraph{Scheduling rule} We run the batches in $F(v)$ in parallel but at different speeds, according to the following scheduling rule. In the process $F(v)$, we keep a counter $y = 1, 2, 3, \cdots$. When the counter is $y$, let $zeros(y)$ be the number of trailing zeros in the binary representation of $y$. Then, for each batch $t\in [1, zeros(y)+1]$, we run one step for each of the recursive calls $H_{t, 1}(v), H_{t, 2}(v), \cdots, H_{t, 2t - 1}(v)$ of batch $t$. That is, we take one step in batch $1$ executions with every $2^0=1$ increase in $y$, one step in batch $2$ executions with every $2^{1}=2$ increase in $y$, and generally one step in batch $t$ executions with every $2^{t-1}$ increase in $y$. 
\end{framed}
\end{minipage}
\end{center}

\medskip


\subsection{Discussion}
\label{subsec:discussion}
Here, we provide a discussion to help understandability, before we proceed to the detailed analysis. We start with the highest level of how the guarantees of the fortified process are tied to the overall comparator output, and then mention the guarantees of the two modules. The exact termination guarantees are nuanced, with carefully controlled error probabilities that depend on each other. 

\medskip
\subsubsection{What the comparator needs from the fortified peeling process} If $\lambda(G) \leq \lambda$, we will show that the expected time bound for the fortified peeling process $F()$ summed up over all sampled vertices is upper bounded by $\frac{S}{\lambda}$ and thus the test will say YES with probability at least $0.9$. On the other hand, suppose $\lambda(G)\geq c\lambda$. Then, the graph includes a core induced subgraph $G[U]$ with minimum degree at least $c\lambda$, and of course where $|U|\geq c\lambda$. In this case, with probability at least $1-(1 - \frac{1}{\lambda})^{c\lambda}\geq 0.85$, the test samples at least one vertex in this core $U$. We will show that for any core vertex, with probability at least $0.8$, the fortified process does not terminate on it. Hence, overall, with probability at least $0.85\cdot 0.8 \geq 0.6$, the processes in the test will not terminate and the test result will be NO.

\subsubsection{Guarantees of the fortified and single-shot peeling processes}
\paragraph{Fortified peeling process} As used above, the fortified peeling process $F(v)$ is devised to come with the following termination guarantees:
\begin{enumerate}
    \item[(I)] If $v \in X$---i.e., if $v$ is a vertex that would get peeled in the sequential greedy algorithm---then $F(v)$ terminates with probability $1$. Moreover, the expected running time of the process in this case is at most $T(v)$. Here $T(v)$ is a deterministic function of the graph, with a recursive definition that will be provided later, and satisfies the following guarantee: $\sum_{v \in V(G)} T(v) = \tilde O(n)$. We set the parameter $S$ in the comparator algorithm, which was so far just mentioned to be $\tilde{O}(n)$, equal to $S= \sum_{v \in V(G)} T(v)$.
  
    \item[(II)] If the graph has a nonempty induced subgraph $G[U]$, where $U \subseteq V$, with minimum degree at least $120 \lambda$, then for any node $u \in U$, with probability at least $0.8$, the process $F(u)$ will never terminate. 
\end{enumerate}

 
\paragraph{Single-shot peeling process} Our fortified process makes use of several calls to the \textit{single-shot} process $H(v)$. This process $H(v)$ comes with certain ``weaker" guarantees:

\begin{enumerate}
    \item[(I)] If $v \in X$, we can define some bad event that occurs during one call of $H(v)$ with probability no more than $0.1$. If this bad event does not occur, the expected running time of $H(v)$ is at most $\frac{1}{5} T(v)$. 
  
    \item[(II)] If the graph has a nonempty induced subgraph  $G[U]$ with minimum degree at least $120 \lambda$, for any $u \in U$, with probability at least $0.9$, the process $H(u)$ will never terminate. 
\end{enumerate}

  \subsection{Analysis}
  \label{subsec:analysis}

  \subsubsection{Low-Arboricity Case}
  We analyze the case of low arboricity graphs with $\lambda(G)\leq \lambda$ here, and show that, in such graphs, with probability at least $0.9$, the result will be YES. That is, all $F(v)$ processes for the sampled vertices $v$ terminate within the total query complexity budget $10S/\lambda$. 
  
  Let us start with a (deterministic) recursive function of the graph and an important property of it. We later use these in the analysis of the algorithm. 
   \begin{lemma}\label[lemma]{lem:recursiveT} Suppose that $G$ is a graph with arboricity at most $\lambda$, and let $X = (x_1, \cdots, x_n)$ be the order of the vertices according to the time of removal in the sequential peeling process. Let $T$ be a function defined recursively as $T(x_i) := \frac{2}{5\lambda} \sum_{j < i, (x_j, x_i) \in E(G)} T({x_j}) + \frac{10 C \, \mathrm{deg}(x_i)}{\lambda}$, where $C$ is a given constant. Then, we have $\sum_{x_i\in X} T(x_i) = O(n)$.
   \end{lemma}
   \begin{proof}
    We call $\mathcal{P}=(x_{p_1}, x_{p_2}, \dots, x_{p_k})$ a \textit{$k$-node directed path} if $p_1 < \cdots < p_k$ and for each two consecutive indices ${p_i}$ and ${p_{i + 1}}$, there is an edge between $x_{p_i}$ and $x_{p_{i + 1}}$ in our graph $G$. Let $Q(k)$ denote the collection of all $k$-node directed paths in graph $G$. 
    
    Let $z_{k, i}$ be the number of $k$-node directed paths $\mathcal{P} \in Q(k)$ that start in node $x_{i}$. In particular, $z_{1, i}=1$ for all $i\in [n]$. Among these, let $z_{k, i, j}$ for $k\geq 2$ and $j>i$ be the number of $k$-node directed paths that start in $x_i$ and have all of their other nodes in $\{x_{j}, \dots, x_{n}\}$. For $k=1$, the start and endpoints have to be the same and thus we define $z_{1, i, j}=1$ for all $i$ and $j$.

    By one step of unrolling $T(x_n)$ with its recursive definition, we have  
    \shortOnly{
    \begin{align*}
        & &&\sum_{i=1}^{n} T(x_i) \\
        &= &&\sum_{i=1}^{n-1} T(x_i) + T(x_n) \\
        &=&&\sum_{i=1}^{n-1} T(x_i) + \frac{2}{5\lambda} \sum_{j < n, (x_j, x_n) \in E(G)} T({x_j}) + \frac{10 C \, \mathrm{deg}(x_n)}{\lambda} \\
        &= &&\sum_{i=1}^{n-1} T(x_i) \left(1+ {z_{2, i, n}} (\frac{2}{5\lambda})\right) + \frac{10 C \, \mathrm{deg}(x_n)}{\lambda}  \\
        &=&&\sum_{i=1}^{n-1} T(x_i) \left(\sum_{k=1}^{2} {z_{k, i, n}} (\frac{2}{5\lambda})^{k-1}\right) + \frac{10 C \, \mathrm{deg}(x_n)}{\lambda}  \\
        &=&&\sum_{i=1}^{n-1} T(x_i) \left(\sum_{k=1}^{\infty} {z_{k, i, n}} (\frac{2}{5\lambda})^{k-1}\right) \\
        & &&+ \frac{10 C \, \mathrm{deg}(x_n)}{\lambda} \cdot \left(\sum_{k=1}^{\infty} z_{k, n} \cdot (\frac{2}{5\lambda})^{k-1}\right),
    \end{align*}
    }
    \fullOnly{
    \begin{align*}
        & &&\sum_{i=1}^{n} T(x_i) \\
        &= &&\sum_{i=1}^{n-1} T(x_i) + T(x_n) \\
        &=&&\sum_{i=1}^{n-1} T(x_i) + \frac{2}{5\lambda} \sum_{j < n, (x_j, x_n) \in E(G)} T({x_j}) + \frac{10 C \, \mathrm{deg}(x_n)}{\lambda} \\
        &= &&\sum_{i=1}^{n-1} T(x_i) \left(1+ {z_{2, i, n}} (\frac{2}{5\lambda})\right) + \frac{10 C \, \mathrm{deg}(x_n)}{\lambda}  \\
        &=&&\sum_{i=1}^{n-1} T(x_i) \left(\sum_{k=1}^{2} {z_{k, i, n}} (\frac{2}{5\lambda})^{k-1}\right) + \frac{10 C \, \mathrm{deg}(x_n)}{\lambda}  \\
        &=&&\sum_{i=1}^{n-1} T(x_i) \left(\sum_{k=1}^{\infty} {z_{k, i, n}} (\frac{2}{5\lambda})^{k-1}\right) + \frac{10 C \, \mathrm{deg}(x_n)}{\lambda} \cdot \left(\sum_{k=1}^{\infty} z_{k, n} \cdot (\frac{2}{5\lambda})^{k-1}\right),
    \end{align*}
    }
    
    where the last equality holds because $z_{k, i, n} = 0$ for $k\geq 3$, $z_{1, n}=1$ and $z_{k, n}=0$ for $k\geq 2$.
    Similarly, with an induction on $\eta$ from $n-1$ to $0$, we can prove that for every $\eta$
    \fullOnly{
    \begin{align*}
        \sum_{i=1}^{n} T(x_i) = \sum_{i=1}^{\eta} \left(T(x_i)\cdot \left(\sum_{k=1}^{\infty} {z_{k, i, \eta+1}} \cdot (\frac{2}{5\lambda})^{k-1}\right)\right) + \sum_{i=\eta+1}^{n} \left(\frac{10 C \, \mathrm{deg}(x_i)}{\lambda} \cdot \left(\sum_{k=1}^{\infty} z_{k, i} \cdot (\frac{2}{5\lambda})^{k-1}\right)\right).
    \end{align*}
    }
    \shortOnly{
    \begin{align*}
        & &&\sum_{i=1}^{n} T(x_i) \\
        &=&& \sum_{i=1}^{\eta} \left(T(x_i)\cdot \left(\sum_{k=1}^{\infty} {z_{k, i, \eta+1}} \cdot (\frac{2}{5\lambda})^{k-1}\right)\right) \\
        & &&+ \sum_{i=\eta+1}^{n} \left(\frac{10 C \, \mathrm{deg}(x_i)}{\lambda} \cdot \left(\sum_{k=1}^{\infty} z_{k, i} \cdot (\frac{2}{5\lambda})^{k-1}\right)\right).
    \end{align*}
    }
   The base case is $\eta=n-1$, which was shown above. For the inductive step, we iteratively remove $T(x_\eta)$ from the left hand side of the above, replacing it with its recursive definition in terms of $T(x_1)$ to $T(x_{\eta-1})$. To analyze this, let us first observe that for any $\eta$ and any $i\leq \eta-1$, we can write 
   \begin{align*}
       z_{k, i, \eta} = z_{k, i, \eta+1} + z_{k-1, \eta, \eta+1} \cdot \mathbf{1}_{(x_i, x_\eta)\in E(G)},
   \end{align*}
    where $\mathbf{1}_{(x_i, x_\eta)\in E(G)}$ is the indicator variable for whether the edge $(x_i, x_\eta)$ is in the graph $G$. Note also that $z_{k, \eta, \eta+1} = z_{k, \eta}$. Using these, for the inductive step, we have:
    \shortOnly{
       \begin{align*}
        & && &&&\sum_{i=1}^{n} T(x_i) \\
        &= && &&&\sum_{i=1}^{\eta} \left(T(x_i)\cdot \left(\sum_{k=1}^{\infty} {z_{k, i, \eta+1}} \cdot (\frac{2}{5\lambda})^{k-1}\right)\right) \\
        & &&+ &&&\sum_{i=\eta+1}^{n} \left(\frac{10 C \, \mathrm{deg}(x_i)}{\lambda} \cdot \left(\sum_{k=1}^{\infty} z_{k, i} \cdot (\frac{2}{5\lambda})^{k-1}\right)\right) \\
        &= && &&&\sum_{i=1}^{\eta-1} \left(T(x_i)\cdot \left(\sum_{k=1}^{\infty} {z_{k, i, \eta+1}} \cdot (\frac{2}{5\lambda})^{k-1}\right)\right) \\
        & &&+ &&&T(x_\eta) \cdot \left(\sum_{k=1}^{\infty} {z_{k, \eta, \eta+1}} \cdot (\frac{2}{5\lambda})^{k-1}\right) \\ 
        & &&+  &&&\sum_{i=\eta+1}^{n} \left(\frac{10 C \, \mathrm{deg}(x_i)}{\lambda} \cdot \left(\sum_{k=1}^{\infty} z_{k, i} \cdot (\frac{2}{5\lambda})^{k-1}\right)\right) \\
        &=&& &&&
        \sum_{i=1}^{\eta-1} \left(T(x_i)\cdot \left(\sum_{k=1}^{\infty} {z_{k, i, \eta+1}} \cdot (\frac{2}{5\lambda})^{k-1}\right)\right) \\
        & &&+  &&&\left(\frac{2}{5\lambda}\cdot\sum_{j < \eta, (x_j, x_\eta) \in E(G)} T({x_j}) + \frac{10 C \, \mathrm{deg}(x_\eta)}{\lambda}\right) \\
        & && &&&\cdot\left(\sum_{k=1}^{\infty} {z_{k, \eta, \eta+1}} \cdot (\frac{2}{5\lambda})^{k-1}\right) \\ 
        & &&+  &&&\sum_{i=\eta+1}^{n} \left(\frac{10 C \, \mathrm{deg}(x_i)}{\lambda} \cdot \left(\sum_{k=1}^{\infty} z_{k, i} \cdot (\frac{2}{5\lambda})^{k-1}\right)\right) \\
        &=&& 
        &&&\sum_{i=1}^{\eta-1} \left(T(x_i)\cdot \left(\sum_{k=1}^{\infty} {z_{k, i, \eta+1}} \cdot (\frac{2}{5\lambda})^{k-1}\right)\right) \\
        & &&+  &&&\left(\frac{2}{5\lambda}\cdot\sum_{j < \eta, (x_j, x_\eta) \in E(G)} T({x_j}) \right) 
        \\
        & && &&&\cdot \left(\sum_{k=1}^{\infty} {z_{k, \eta, \eta+1}} \cdot (\frac{2}{5\lambda})^{k-1}\right) \\ 
        & &&+ &&&\left(\frac{10 C \, \mathrm{deg}(x_\eta)}{\lambda}\right)\cdot \left(\sum_{k=1}^{\infty} {z_{k, \eta}} \cdot (\frac{2}{5\lambda})^{k-1}\right) \\
        & &&+ &&& \sum_{i=\eta+1}^{n} \left(\frac{10 C \, \mathrm{deg}(x_i)}{\lambda} \cdot \left(\sum_{k=1}^{\infty} z_{k, i} \cdot (\frac{2}{5\lambda})^{k-1}\right)\right) \\
        & =&& &&&\sum_{i=1}^{\eta-1} \left(T(x_i)\cdot \left(\sum_{k=1}^{\infty} {z_{k, i, \eta}} \cdot (\frac{2}{5\lambda})^{k-1}\right)\right) \\ & &&+ &&&\sum_{i=\eta}^{n} \left(\frac{10 C \, \mathrm{deg}(x_i)}{\lambda} \cdot \left(\sum_{k=1}^{\infty} z_{k, i} \cdot (\frac{2}{5\lambda})^{k-1}\right)\right),
    \end{align*}
    }
    \fullOnly{
       \begin{align*}
        & && \sum_{i=1}^{n} T(x_i) \\
        &= && \sum_{i=1}^{\eta} \left(T(x_i)\cdot \left(\sum_{k=1}^{\infty} {z_{k, i, \eta+1}} \cdot (\frac{2}{5\lambda})^{k-1}\right)\right) 
       + \sum_{i=\eta+1}^{n} \left(\frac{10 C \, \mathrm{deg}(x_i)}{\lambda} \cdot \left(\sum_{k=1}^{\infty} z_{k, i} \cdot (\frac{2}{5\lambda})^{k-1}\right)\right) \\
        &= && \sum_{i=1}^{\eta-1} \left(T(x_i)\cdot \left(\sum_{k=1}^{\infty} {z_{k, i, \eta+1}} \cdot (\frac{2}{5\lambda})^{k-1}\right)\right) \\
        & &&+ T(x_\eta) \cdot \left(\sum_{k=1}^{\infty} {z_{k, \eta, \eta+1}} \cdot (\frac{2}{5\lambda})^{k-1}\right) \\ 
        & &&+  \sum_{i=\eta+1}^{n} \left(\frac{10 C \, \mathrm{deg}(x_i)}{\lambda} \cdot \left(\sum_{k=1}^{\infty} z_{k, i} \cdot (\frac{2}{5\lambda})^{k-1}\right)\right) \\
        &=&& 
        \sum_{i=1}^{\eta-1} \left(T(x_i)\cdot \left(\sum_{k=1}^{\infty} {z_{k, i, \eta+1}} \cdot (\frac{2}{5\lambda})^{k-1}\right)\right) \\
        & &&+  \left(\frac{2}{5\lambda}\cdot\sum_{j < \eta, (x_j, x_\eta) \in E(G)} T({x_j}) + \frac{10 C \, \mathrm{deg}(x_\eta)}{\lambda}\right) 
        \cdot\left(\sum_{k=1}^{\infty} {z_{k, \eta, \eta+1}} \cdot (\frac{2}{5\lambda})^{k-1}\right) \\ 
        & &&+  \sum_{i=\eta+1}^{n} \left(\frac{10 C \, \mathrm{deg}(x_i)}{\lambda} \cdot \left(\sum_{k=1}^{\infty} z_{k, i} \cdot (\frac{2}{5\lambda})^{k-1}\right)\right) \\
        &=&& 
        \sum_{i=1}^{\eta-1} \left(T(x_i)\cdot \left(\sum_{k=1}^{\infty} {z_{k, i, \eta+1}} \cdot (\frac{2}{5\lambda})^{k-1}\right)\right) \\
        & &&+ \left(\frac{2}{5\lambda}\cdot\sum_{j < \eta, (x_j, x_\eta) \in E(G)} T({x_j}) \right) 
        \cdot \left(\sum_{k=1}^{\infty} {z_{k, \eta, \eta+1}} \cdot (\frac{2}{5\lambda})^{k-1}\right) \\ 
        & &&+ \left(\frac{10 C \, \mathrm{deg}(x_\eta)}{\lambda}\right)\cdot \left(\sum_{k=1}^{\infty} {z_{k, \eta}} \cdot (\frac{2}{5\lambda})^{k-1}\right) \\
        & &&+  \sum_{i=\eta+1}^{n} \left(\frac{10 C \, \mathrm{deg}(x_i)}{\lambda} \cdot \left(\sum_{k=1}^{\infty} z_{k, i} \cdot (\frac{2}{5\lambda})^{k-1}\right)\right) \\
        & =&& \sum_{i=1}^{\eta-1} \left(T(x_i)\cdot \left(\sum_{k=1}^{\infty} {z_{k, i, \eta}} \cdot (\frac{2}{5\lambda})^{k-1}\right)\right) + \sum_{i=\eta}^{n} \left(\frac{10 C \, \mathrm{deg}(x_i)}{\lambda} \cdot \left(\sum_{k=1}^{\infty} z_{k, i} \cdot (\frac{2}{5\lambda})^{k-1}\right)\right),
    \end{align*}
    }
    which completes the induction proof.
    
    \bigskip
    \noindent From this now proven equality 
    \shortOnly{
    \begin{align*}
        & &&\sum_{i=1}^{n} T(x_i) \\
        & = &&\sum_{i=1}^{\eta} \left(T(x_i)\cdot \left(\sum_{k=1}^{\infty} {z_{k, i, \eta+1}} \cdot (\frac{2}{5\lambda})^{k-1}\right)\right) \\ & &&+ \sum_{i=\eta+1}^{n} \left(\frac{10 C \, \mathrm{deg}(x_i)}{\lambda} \cdot \left(\sum_{k=1}^{\infty} z_{k, i} \cdot (\frac{2}{5\lambda})^{k-1}\right)\right),
    \end{align*}
    }
    \fullOnly{
    \begin{align*}
        \sum_{i=1}^{n} T(x_i) 
        = \sum_{i=1}^{\eta} \left(T(x_i)\cdot \left(\sum_{k=1}^{\infty} {z_{k, i, \eta+1}} \cdot (\frac{2}{5\lambda})^{k-1}\right)\right) + \sum_{i=\eta+1}^{n} \left(\frac{10 C \, \mathrm{deg}(x_i)}{\lambda} \cdot \left(\sum_{k=1}^{\infty} z_{k, i} \cdot (\frac{2}{5\lambda})^{k-1}\right)\right),
    \end{align*}
    }
    by setting $\eta=0$, we get 
     \begin{align*}
        \sum_{i=1}^{n} T(x_i) = \sum_{i=1}^{n} \left(\frac{10 C \, \mathrm{deg}(x_i)}{\lambda} \cdot \left(\sum_{k=1}^{\infty} z_{k, i} \cdot (\frac{2}{5\lambda})^{k-1}\right)\right).
    \end{align*}
    %
    %
    We next analyze this summation. We use the fact that $z_{k, i}\leq (2\lambda)^{k-1}$ for each $i$ and $k$, since each node $x_j$ has at most $2\lambda$ neighbors $x_{j'}$ for $j'>j$. Thus we have:

    \begin{align*}
        \sum_{i=1}^{n} T(x_i) &= && \sum_{i=1}^{n} \left(\frac{10 C \, \mathrm{deg}(x_n)}{\lambda} \cdot \left(\sum_{k=1}^{\infty} z_{k, i} \cdot (\frac{2}{5\lambda})^{k-1}\right)\right)\\
        &\leq &&\sum_{i=1}^{n} \frac{10 C \mathrm{deg}(x_i)}{\lambda} \sum_{k=1}^{\infty} (\frac{2}{5\lambda})^{k - 1} (2\lambda)^{k - 1} \\
        &= &&\sum_{i=1}^{n} \frac{10 C \mathrm{deg}(x_i)}{\lambda} \sum_{k=1}^{\infty} (\frac{4}{5})^{k - 1} \\
    &= &&\sum_{i=1}^{n} \frac{50C \mathrm{deg}(x_i)}{\lambda} \\
    &\leq&& \frac{50C}{\lambda} (2n\lambda)= O(n) 
    \end{align*}
    

   \end{proof}

  \begin{lemma}
  \label[lemma]{lem:smallCase}
  Suppose that $G$ is a graph with arboricity at most $\lambda$, and let $X = (x_1, \cdots, x_n)$ be the order of the vertices according to the time of removal in the sequential peeling process. Then, when we call the fortified process $F(x_i)$ on each node $x_i$, the process terminates with probability $1$ and in expected time at most $T(x_i)$, where $T()$ is defined in \Cref{lem:recursiveT} and satisfies $\sum_{i \in X} T(x_i) = O(n)$.        
  \end{lemma}
  \begin{proof} The proof has two parts, which we present separately.
  
  \paragraph{Part I} First, let us consider one call of the simpler single-shot process $H(x_i)$. Let us say $H(x_i)$ has a \textit{bad sample} if there exists at least one sampled neighbor $u$ of $x_i$ such that $u\notin \bigcup_{j=1}^{i-1} \{x_j\}$. Since $x_i$ has at most $2\lambda$ neighbors that are not in $\bigcup_{j=1}^{i-1} \{x_j\}$, the probability of a bad sample is at most $\frac{2\lambda}{25\lambda} \leq 0.1$. 
  Let $Z(x_i) := \mathbb{E}[\textit{runtime of } H(x_i) | H(x_i)\textit{does not have a bad sample}]$. We argue that $$\mathbb{E}[\textit{runtime of } F(x_i)] \leq 10 Z(x_i).$$
  In particular, we bound the expected running time of $F(x_i)$ by considering the smallest batch index $p$ such that less than half of the recursive calls $H_{p,.}(x_i)$ in that batch have a bad sampling. Since each call is bad with probability at most $0.1$, the probability that $p = t$ is at most $0.1^{t - 1}$. On the other hand, the expected running time of $F$ conditioned on $p = t$ is upper bounded by 
  \begin{align*} (t 2^{t - 1} \sum_{j = 1}^{\infty} \frac{2j - 1}{2^{j-1}}) \cdot Z(x_i) = (3t \cdot 2^t) \cdot Z(x
_i).
  \end{align*} This is because of our scheduling rule and regardless of the random samples in batches $1$ to $p-1$. Concretely, the expected total running time of $t$ single-shot peeling processes $H(x_i)$, conditioned on them not having bad samples, is $t Z(x_i)$. Thus, the expected value of the counter $y$ when batch $p$ terminates is upper bounded by $t2^{t-1} Z(x_i)$. At the counter value $y$, the number of steps we have taken in the $2j-1$ calls of batch $j$, for each $j$, is at most $y \cdot (2j-1)/2^{j-1}$. Thus, over all $j$, the number of steps is at most $(t 2^{t - 1} \sum_{j = 1}^{\infty} \frac{2j - 1}{2^{j-1}}) \cdot Z(x_i) = (3t \cdot 2^t) \cdot Z(x_i)$. Using this and that the probability of $p=t$ is at most $(0.1)^{t-1}$, we conclude
  \begin{align}
  \label{lem:ineqFtoZ}
        \mathbb{E}[\textit{runtime of } F(x_i)] \leq \sum_{t=1}^{\infty} (0.1)^{t-1} (3t \cdot 2^t) \cdot Z(x_i) \leq 10 Z(x_i). \tag{**}
  \end{align}

    \paragraph{Part II} Using part I and concretely \Cref{lem:ineqFtoZ}, and with the help of an induction on $i$ (in particular for node $x_i\in (x_1, \cdots, x_n)$), we argue that the fortified process $F(x_i)$ terminates with probability $1$ and $\mathbb{E}[\textit{runtime of } F(x_i)] \leq T(x_i),$ for the recursive function $T()$ defined in \Cref{lem:recursiveT}.
    
    Notice that we have 
  \begin{align*} & &&Z(x_i)\\
  & := &&\mathbb{E}[\textit{runtime of } H(x_i) | H(x_i)\textit{doesn't have a bad sample}] \\ 
  &\leq &&\frac{1}{25 \lambda} \sum_{j < i, (x_j, x_i) \in E(G)} \mathbb{E}[\textit{runtime of }F(x_j)] + \frac{C\, \mathrm{deg}(x_i)}{\lambda}.
  \end{align*} The first term with summation over $j<i$ is because we make a recursive fortified process call $F(x_j)$ on each neighbor $j<i$ of $x_i$ with probability $1/(25\lambda)$, and no call on $x_{j'}$ for $j'\geq i$ given the condition that $H(x_i)$ has a good sample.
  The second term $\frac{C\,\mathrm{deg}(x_i)}{\lambda}$ is the expected query complexity to sample the neighbors of $x_i$. Furthermore, $\frac{C\,\mathrm{deg}(x_i)}{\lambda}$ is also an upper bound on the expected time complexity to sample the neighbors, using an efficient sampling implementation. See \Cref{lem:sampling} for the details of such an efficient sampling implementation. 
  
  For the base case of $x_1$, as a result of the above and \Cref{lem:ineqFtoZ}, we have
  $\mathbb{E}[\textit{runtime of } F(x_i)] \leq 10Z(x_1) \leq 10 C deg(x_1)/\lambda = T(x_1)$. This proves the base case of our induction. For the inductive step, let us assume that $\mathbb{E}[\textit{runtime of }F(x_j)]\leq T(x_j)$ for all $j<i$. Using \Cref{lem:ineqFtoZ} again, we conclude
    \shortOnly{
    \begin{align*}
    & &&\mathbb{E}[\textit{runtime of } F(x_i)] \\
    &\leq && 10Z(x_i) \\ 
    &\leq &&10 \Big(\frac{1}{25 \lambda} \sum_{j < i, (x_j, x_i) \in E(G)} \mathbb{E}[\textit{runtime of }F(x_j)] \\
    & &&+ \frac{C\, \mathrm{deg}(x_i)}{\lambda}\Big) \\
    &\leq &&10 \big(\frac{1}{25 \lambda} \sum_{j < i, (x_j, x_i) \in E(G)} T(x_j) + \frac{C\,\mathrm{deg}(x_i)}{\lambda} \big) \\ 
    &= &&\frac{2}{5\lambda} \sum_{j < i, (x_j, x_i) \in E(G)} T({x_j}) + \frac{10C \,\mathrm{deg}(x_i)}{\lambda}  = T(x_i), 
    \end{align*}
    }
    \fullOnly{
     \begin{align*}
    & &&\mathbb{E}[\textit{runtime of } F(x_i)] \\
    &\leq && 10Z(x_i) \\ 
    &\leq &&10 \Big(\frac{1}{25 \lambda} \sum_{j < i, (x_j, x_i) \in E(G)} \mathbb{E}[\textit{runtime of }F(x_j)] 
    + \frac{C\, \mathrm{deg}(x_i)}{\lambda}\Big) \\
    &\leq &&10 \big(\frac{1}{25 \lambda} \sum_{j < i, (x_j, x_i) \in E(G)} T(x_j) + \frac{C\,\mathrm{deg}(x_i)}{\lambda} \big) \\ 
    &= &&\frac{2}{5\lambda} \sum_{j < i, (x_j, x_i) \in E(G)} T({x_j}) + \frac{10C \,\mathrm{deg}(x_i)}{\lambda} \\
    &=&& T(x_i), 
    \end{align*}
    }
    where the penultimate inequality used the inductive hypothesis. This concludes the proof.
\end{proof}

\begin{corollary}\label[corollary]{crl:smallCase}
If $\lambda(G)\leq \lambda$, then in each test, with probability at least $0.9$, the result will be YES, i.e., all $F(v)$ processes for the sampled vertices $v$ terminate within query complexity budget $10S/\lambda$.    
\end{corollary}
\begin{proof}
    By \Cref{lem:smallCase}, in this case each call to $F(v)$ terminates with probability $1$. Furthermore, since we sample each node with probability $1/\lambda$, and as the process on vertex $x_i$ terminates in expected time at most $T(x_i)$, for which we have $S=\sum_{i \in X} T(x_i) = O(n)$, the expected time (and thus also query complexity) for all calls to $F()$ to terminate is at most $S/\lambda$. Hence, with probability at least $0.9$, all the calls terminate within the budget of $10S/\lambda$.
\end{proof}   

\subsubsection{High-Arboricity Case}
We argue here that if $\lambda(G)\geq 120 \lambda$, then in each test, with probability at least $0.6$, the result will be NO, i.e., some processes will not terminate. The main ingredient is the following lemma.
\begin{lemma}
\label[lemma]{lem:largeCase}
 Suppose that $G$ is a graph with arboricity at least $120\lambda$. Let $G[U]$ be an induced subgraph with a minimum degree at least $120 \lambda$. For any $v \in U$,  with probability at least $0.8$, the fortified process $F(v)$ will never terminate. 
 \end{lemma} 
 \begin{proof}

 Starting from one call of $H(v)$, the recursive calls of $F$ and $H$ can be drawn as a recursion tree with one node for each call connected to the other call that invoked it as its parent. The initial call (i.e., the one represented by the root of the tree) terminates if and only if the tree has a finite height. We will prove by induction on the tree height $h$ that for each $v \in U$, (A) the probability that running $F(v)$ ends with a tree of height at most $h$ is at most $0.2$, and (B) the probability that running $H(v)$ ends with a tree of height at most $h$ is at most $0.1$. The base case of $h = 0$ is trivial for both claims. Suppose $h > 0$. We examine claims (A) and (B) separately. 

 \begin{itemize}
     \item (A) For a run of the fortified $F(v)$, it terminates within height $h$ if and only if the majority of recursive calls to $H_{t,.}(v)$ in one of the batches $t$ terminate within height $h - 1$. By the induction assumption and a union bound, the probability that this happens for a certain batch $t$ is at most $2^{2t - 2}0.1^t$, because there are $2^{2t - 2}$ combinations of at least $t$ calls and each combination has probability at most $0.1^t$ to terminate. With a union bound over different values of $t$, we conclude that the probability of $F(v)$ terminating with a tree of height at most $h$ is at most $\sum_{t = 1}^{\infty} 2^{2t - 2}0.1^t < 0.2$.

 \item (B) In the process $H(v)$ run on node $v$, each of its neighbors $u\in N(v) \cap U$, of which there are at least $120\lambda$ many, is sampled for running $F(u)$ with probability $1/25\lambda$. The probability that strictly less than two neighbors are sampled is at most 
 \shortOnly{
 \begin{align*} & &&(1-\frac{1}{25\lambda})^{120\lambda} + 120\lambda \cdot \frac{1}{25\lambda} \cdot (1-\frac{1}{25\lambda})^{119\lambda} \\
 &< && 6 (1-\frac{1}{25\lambda})^{119\lambda} \leq 6 e^{-119/25} < 0.06.
 \end{align*}
 }\fullOnly{
 \begin{align*} (1-\frac{1}{25\lambda})^{120\lambda} + 120\lambda \cdot \frac{1}{25\lambda} \cdot (1-\frac{1}{25\lambda})^{119\lambda} 
 <  6 (1-\frac{1}{25\lambda})^{119\lambda} \leq 6 e^{-119/25} < 0.06.
 \end{align*}
 }
 Thus, with probability at least $0.94$, at least $2$ neighbors are sampled. The probability that both of these sampled neighbors have trees of height at most $h-1$ for their $F(u)$ process is at most $(0.2)^2=0.04$, by induction, and their independence. Hence, the probability of $H(v)$ terminating with a tree of height at most $h$ is at most $0.04+0.06 = 0.1$.
 \end{itemize}
 \end{proof}

 \begin{corollary}
    \label[corollary]{crl:largeCase}
    If $\lambda(G)\geq 120 \lambda$, then in each test, with probability at least $0.6$, the result will be NO, i.e., some processes will not terminate.
\end{corollary}
\begin{proof} With probability at least $1-(1 - \frac{1}{\lambda})^{120\lambda}\geq 0.85$, the test samples at least one vertex in the core $U$. By \Cref{lem:largeCase}, for such a sampled core vertex, the fortified process does not terminate on it with probability at least $0.8$. Hence, overall, with probability at least $0.85\cdot 0.8 \geq 0.6$, the process will not terminate and the test result will be NO. 
\end{proof}

\subsubsection{Wrap Up}
\label{subsub:wrap}
\begin{corollary}
If $\lambda(G)\leq \lambda$, the comparator will say YES with high probability, and if $\lambda(G)>120\lambda$, it will say NO, with high probability. 
\end{corollary}
\begin{proof}
    The comparator consists of $\Theta(\log n)$ independent tests and outputs the majority result. If $\lambda(G)\leq \lambda$, by \Cref{crl:smallCase} each test will say YES with probability at least $0.9$ and thus their majority will say YES with high probability. If $\lambda(G)>120\lambda$, by \Cref{crl:largeCase} each test will say NO with probability at least $0.6$ and thus their majority will say NO with high probability.
\end{proof} 

\begin{proof}[Proof of \Cref{thm:Main}]
The comparator discussed above runs using $\tilde{O}(n/\lambda)$ time and query complexity, with the guarantee that, with high probability, if $\lambda(G)\leq \lambda$, its says YES and if $\lambda(G)>120\lambda$, it say NO. Given this comparator, the $O(1)$ approximation algorithm claimed in \Cref{thm:Main} follows, in $\tilde{O}(n/\lambda(G))$ time and queries, as standard and mentioned in the preliminaries of \Cref{sec:overview}: we just run the comparator for thresholds $\lambda=n/2^i$ for $i=0, 1, 2, \ldots$ until we reach the first value of $\lambda$ that the comparator outputs YES.
\end{proof}

\bibliographystyle{alpha}
\bibliography{refs}

\end{document}